\documentclass[sigconf,edbt]{acmart-edbt2021}

\def\BibTeX{{\rm B\kern-.05em{\sc i\kern-.025em b}\kern-.08em
    T\kern-.1667em\lower.7ex\hbox{E}\kern-.125emX}}

\usepackage{booktabs} 

\usepackage{graphicx}
\usepackage[caption=false]{subfig}
\usepackage{algorithm}
\usepackage{algorithmic}
\usepackage{balance}

\usepackage{enumitem} 

\newcommand{\oneLineIF}[2]{\STATE \algorithmicif\ #1 \algorithmicthen\ #2} 
\newcommand{\oneLineIFandElse}[3]{\STATE \algorithmicif\ #1 \algorithmicthen\ #2 \algorithmicelse\ #3} 
\newcommand{\paddingT}{\vskip -0.125in}
\newcommand{\paddingD}{\vskip -0.20in}

\usepackage{sty/term}       
\usepackage{sty/local}      

\setcopyright{rightsretained}

\acmDOI{}

\acmISBN{978-3-89318-084-4}

\acmConference[EDBT 2021]{24th International Conference on Extending Database Technology (EDBT)}{March 23-26, 2021}{Nicosia, Cyprus} 
\acmYear{2021}

\begin{document}
\title{Efficient Discovery of Approximate Order Dependencies}

\author{Reza Karegar}
\affiliation{%
  \institution{University of Waterloo, CA}
}
\email{mkaregar@uwaterloo.ca}

\author{Parke Godfrey}
\affiliation{%
  \institution{York University, CA}
}
\email{godfrey@yorku.ca}

\author{Lukasz Golab}
\affiliation{%
  \institution{University of Waterloo, CA}
}
\email{lgolab@uwaterloo.ca}

\author{Mehdi Kargar}
\affiliation{%
  \institution{Ryerson University, CA}
}
\email{kargar@ryerson.ca}

\author{Divesh Srivastava}
\affiliation{%
  \institution{AT\&T Labs-Research, US}
}
\email{divesh@research.att.com}

\author{Jaroslaw Szlichta}
\affiliation{%
  \institution{Ontario Tech Univ, CA}
}
\email{jarek@ontariotechu.ca}


\renewcommand{\shortauthors}{}

\begin{abstract}

Order dependencies (\OD s) capture relationships
between ordered domains of attributes.
Approximate \OD s (\AOD s) capture such relationships
even when there exist exceptions in the data.
During automated discovery of \OD s,
\emph{validation} is the process
of verifying whether an $\OD$ holds.
We present an algorithm
for validating approximate \OD s
with significantly improved runtime performance
over existing methods for \AOD s,
and prove that it is \emph{correct} and has \emph{optimal} runtime. 
By replacing the validation step
in a leading algorithm
for approximate \OD\ discovery
with ours,
we achieve orders-of-magnitude improvements in performance.

\end{abstract}

\maketitle

\section{Introduction} \label{sec:intro}

\subsection{Motivation}
\emph{Functional dependencies} (\FD s) specify that the values of given attributes
\emph{functionally determine} the value
of a target attribute.
\emph{Order dependencies} extend \FD s
to state that, additionally,
the \emph{order} of tuples
with respect to the values from the domains of given attributes
determines the \emph{order}
of the values from the domain
of the target attribute.
Table \ref{table:taxes} shows a dataset with employee salaries. 
In this table,
the \OD\ that \A{sal}\ \emph{orders} \A{taxGrp}\ holds.
If one sorts the table by \A{sal},
it is sorted by \A{taxGrp}\ as well. 

An \OD\ implies the corresponding \FD; 
e.g., that \A{sal}\ orders \A{taxGrp}\ implies that 
\A{sal}\ \emph{functionally determines} \A{taxGrp}. 
\emph{Order compatibility} (\OC) captures the \emph{co-ordering}
aspect of an \OD\ 
\emph{without} the corresponding \FD. 
Any \OD\ can be equivalently represented
by a pair of an \OC\ and an \FD\
\cite{od-fundamentals2012}.
In Table \ref{table:taxes},
that \A{taxGrp}\ is \emph{order compatible} with \A{sal}\ holds
as there exists an order of the tuples
such that they are sorted both
by \A{taxGrp}\
and
by \A{sal}.
Note that $\A{taxGrp}$ does \emph{not} order $\A{sal}$ as an $\FD$ does not hold.

There has been recent work
to automate the \emph{discovery} of \OD s from data
\cite{fastod2017,fastod-vldbj,
	od-discovery-consonni2019,
	od-discovery-jinand2020,ORDER}.
In practice,
however,
such constraints rarely hold \emph{perfectly}
in the data.
Real data are dirty,
containing wrong and inconsistent values
that may violate semantically valid dependencies.
This motivates the need
for discovering \emph{approximate} \OD s (\AOD s), 
\OD s that hold in the data but with \emph{exceptions}. 
Discovered \OD s deemed semantically valid
can be used for data cleaning,
to detect erroneous tuples,
where measures are then taken to repair the errors~\cite{QTY18}.
Approximate \OD s are useful
even when the data are not dirty, 
as there can be exceptions to general rules.
\AOD s help avoid overfitting
by discovering more general and meaningful dependencies.

In Table~\ref{table:taxes},
tax is a fixed percentage of salary 
in each tax group;
i.e., one, three, or eight percent.  However, \A{perc}\ includes a concatenated zero in some rows due to data entry errors
(e.g., 10\% instead of 1\% in $\tup{t}_1$).
Because of this,
the \OC\ that \A{salary}\ is order compatible with \A{tax}\
does \emph{not} hold, 
even though this \OC\ is intended.
Similarly,
the \FD\ that \A{pos}, \A{exp} \ functionally determines \A{sal}\
does \emph{not} hold,
due to the exception
of tuples $\tup{t}_6$\ and $\tup{t}_7$,
two employees with the same position and years of experience
but having different salaries. 
With approximate \OD s,
we can still discover such concise and meaningful rules
in these instances. 

\begin{table}[t]
\begin{center}
\caption{Employee salaries}
\paddingT
\label{table:taxes}
\fontsize{8}{5}\selectfont
    \begin{tabular}{|l|l|l|l|l|l|l|l|l}
        \hline
        \A{\#} & 
        \A{pos} & \A{exp} & \A{sal} & \A{taxGrp} & \A{perc} & \A{tax} & \A{bonus} \Tstrut\\
        \hline \hline
        $\tup{t}_1$ & 
        sec & 1 & 20K & A &  10\% &  2K & 1K \Tstrut\\
        \hline
        $\tup{t}_2$ & 
        sec &  3 & 25K & A & 10\% &  2.5K & 1K \Tstrut\\
        \hline
        $\tup{t}_3$ & 
        dev &  1 & 30K & A & 1\% &  0.3K & 3K \Tstrut\\
        \hline
        \hline
        $\tup{t}_4$ & 
        sec &  5 & 40K & B & 30\% &  12K & 2K \Tstrut\\
        \hline
        $\tup{t}_5$ & 
        dev &  3 & 50K & B & 3\% &  1.5K & 4K \Tstrut\\
        \hline
        $\tup{t}_6$ & 
        dev &  5 & 55K & B & 30\% &  16.5K & 4K \Tstrut\\
        \hline
        $\tup{t}_7$ & 
        dev &  5 & 60K & B & 3\% &  1.8K & 4K \Tstrut\\
        \hline
        \hline
        $\tup{t}_8$ & 
        dev &  -1 & 90K & C & 8\% &  7.2K & 7K \Tstrut\\
        \hline
        $\tup{t}_9$ & 
        dir &  8 & 200K & C & 8\% &  16K & 10K \Tstrut\\
        \hline
    \end{tabular}
\end{center}
\paddingD
\end{table}

Approximate canonical \OD s were first introduced 
in \cite{fastod2017}. 
Their definition of \AOD s,
as is ours herein,
is based on the concept of ``tuple removal\@.''
Given a table and an \OD,
a \emph{removal set} is a set of tuples
which,
if removed from the table,
results in the \OD\ holding.
A \emph{minimal} removal set is
one with the smallest cardinality.
An \emph{approximation factor} can be defined
with respect to a table and an \OD,
as the ratio of the size of a minimal removal set
over the size of the table.
For instance,
for Table \ref{table:taxes} and the \OC\
that \A{pos}, \A{exp}\ is order compatible with \A{pos}, \A{sal},
the minimal removal set and the approximation factor are
$\{\tup{t}_8\}$
and
$1/9\approx0.11$,
respectively.

Given a table \T{r}\ and an approximation threshold $0\le \epsilon\le1$, 
the \emph{discovery problem} for \AOD s
is to find the complete set of minimal \emph{valid} \AOD s
in \T{r}\ w.r.t.\ $\epsilon$.
Exact \OD s are a special case of \AOD s
with an approximation factor of zero.
Given a table \T{r}, an \OD\ $\varphi$,
and a threshold $\epsilon$, 
the problem of 
\emph{validating} the candidate \OD\ as an \AOD\ involves verifying 
whether the approximation factor of $\varphi$,
denoted by $e(\varphi)$,
is less than or equal to $\epsilon$.

\subsection{Contributions}
The extension for \AOD\ discovery
in \cite{fastod2017,fastod-vldbj},
however,
is impractical due to its performance.
%
To validate a candidate $\AOC$ in the search,
they iteratively remove the tuple%
---%
or one of the tuples,
in the case of a tie%
---%
that causes the largest number of violations.
This has two weaknesses: the runtime is quadratic in the number of tuples, and it is not guaranteed to find a minimal removal set.


That it is quadratic makes
it prohibitively expensive to run on larger datasets. 
(The validation step for a candidate exact \OD\ has
a linear runtime in the number of tuples.)
So while the \OD\ discovery algorithm
in \cite{fastod2017,fastod-vldbj}
is shown to scale
to datasets with millions of tuples,
it is infeasible to run their adapted \AOC\ discovery algorithm
over even moderately sized datasets.
During benchmarking,
we found in some discovery runs
that more than 99\% of the running time
is spent on validating $\AOC$ candidates. 

That it deliberately does not guarantee finding a minimal removal set means
that the algorithm may overestimate 
the approximation factor of an \AOC\ candidate.
Thus, \emph{true} \AOC s
with respect to the approximation threshold
can be eliminated,
which means that the \AOC\ discovery algorithm is \emph{incomplete} (while the exact \OD\ discovery algorithm is complete).

In this paper,
we resolve this major bottleneck in \AOD\ discovery 
via an algorithm with optimal runtime
and guaranteed minimal removal set
for validating \AOC\ candidates.
This brings performance of \AOD\ discovery
on par with that of \OD\ discovery
while making the \AOD\ discovery complete.

The paper is structured as follows,
with the following key contributions.
In Section \ref{sec:prelim},
we provide background and discuss related work. 
In Sec. \ref{sec:approx:framework}
we illustrate the established \OD\ and \AOD\ discovery framework%
---%
which we then adapt herein%
---
and, in Sec. \ref{sec:approx:iterative},
the \emph{iterative validation algorithm}~\cite{fastod2017,fastod-vldbj} it employs.
In Sec. \ref{sec:approx:optimal},
we contribute a minimal and optimal validation algorithm based on longest increasing subsequences that decreases the runtime from quadratic to super-linear.
In Sec. \ref{sec:exp},
we present our experimental results,
with the following contributions.
%
%
%
			We demonstrate that \AOD\ discovery
			using our validation algorithm 
			scales to datasets with millions of tuples
			and tens of attributes
			(Exp-\ref{exp:scal-tup} and Exp-\ref{exp:scal-attr}).
			We compare our adapted \AOC\ discovery
			against the previous approach 
			and demonstrate that ours is orders of magnitude faster
			(Exp-\ref{exp:vs-iter-time}).
			As discovering \AOD s enables
			the application of pruning rules earlier
			than for discovering \OD s,
			\AOD\ discovery can be just as efficient,
			if not more so.
			%
			Our \AOD\ discovery algorithm gains up to 76\% improvement
			in runtime 
			compared against the (exact) \OD\ discovery algorithm
			(Exp-\ref{exp:vs-exact-time}).
			Given our \AOD\ discovery algorithm is complete,
			we discover more \AOD s,
			and \emph{semantically more general} \AOD s
			(thus, of higher quality).
			We show that we find more \AOD s, 
			both due to our better scalability
			and the minimality of our removal sets
			(Exp-\ref{exp:vs-iter-aoc} and Exp-\ref{exp:vs-exact-aoc}).
			%
In Section \ref{sec:conclusions},
we conclude with suggestions for future work. 
In Section~\ref{sec:appendix}, we prove our theorems.

\section{Preliminaries and Related Work} \label{sec:prelim}

\subsection{Definitions and Notation} \label{sec:prelim:notation}
$\R{R}$ denotes a relational schema,
$\T{r}$ represents a table instance, and 
$\tup{s}$ and $\tup{t}$ denote tuples.
$\A{A}$ and $\A{B}$ denote individual attributes and 
$\set{X}$ and $\set{Y}$ \emph{sets} of attributes. 
\emph{Lists} of attributes are presented using $\lst{X}$ and $\lst{Y}$; 
$\emptyLst$ denotes the empty list and 
\(\Lst{\A{A}}[\lst{T}]\) denotes a list
with \emph{head} attribute \(\A{A}\)
and
\emph{tail} list \(\lst{T}\).
Tuples 
$\proj{t}{\A{A}}$ and $\proj{t}{\set{X}}$ denote the projections of tuple $\tup{t}$
on $\A{A}$ and $\set{X}$, respectively. 
Wherever a set is expected but a list appears,
the list is cast to a set;
e.g., \(\proj{t}{\lst{X}}\) is equivalent
to \(\proj{t}{\set{X}}\).
$\lst{X}'$ represents an arbitrary
permutation of the values of a list $\lst{X}$ or set $\set{X}$.

\begin{definition} \label{def:nested-order}
\emph{(nested order)}
Let \(\lst{X}\) be a list of attributes
where \(\set{X} \in \R{R}\).
Given two tuples, $\tup{s}$ and $\tup{t}$,
\(\tup{s}\orel[\lst{X}]\tup{t}\) \emph{iff}

\begin{itemize}[nolistsep]
\item
    \(\lst{X} = \Lst{}\);
    {\em or}
\item
    \(\lst{X} = \Lst{\A{A}}[\lst{T}]\)
    and \(\tup{s}_{\A{A}} < \tup{t}_{\A{A}}\); or
\item
    \(\lst{X} = \Lst{\A{A}}[\lst{T}]\),
    \(\tup{s}_{\A{A}} = \tup{t}_{\A{A}}\),
    and
    \(\tup{s}\orel[\lst{T}]\tup{t}\).
\end{itemize}

Let \(\tup{s}\orelStrict[\lst{X}]\tup{t}\)
\emph{iff}
    \(\tup{s}\orel[\lst{X}]\tup{t}\) but
    \(\tup{t}\notOrel[\lst{X}]\tup{s}\).
\end{definition}

Next, we define order dependencies~\cite{od-fundamentals2012,fastod2017,fastod-vldbj, od-discovery-consonni2019, od-discovery-jinand2020,ORDER}.

\begin{definition} \label{def:order-dep}
\emph{(order dependency)}
Let $\lst{X}$ and $\lst{Y}$ be lists of attributes where
$\set{X}, \set{Y}\subseteq\R{R}$. 
$\orders{\lst{X}}{\lst{Y}}$ denotes an \emph{order dependency}, 
read as $\lst{X}$ \emph{orders} $\lst{Y}$. 
Table $\T{r}$ satisfies $\orders{\lst{X}}{\lst{Y}}$ 
($\T{r} \models \orders{\lst{X}}{\lst{Y}}$) iff,
for all $\tup{s}, \tup{t}\in\T{r}$, 
\(\tup{s}\orel[\lst{X}]\tup{t}\) implies \(\tup{s}\orel[\lst{Y}]\tup{t}\).
$\lst{X}$ and $\lst{Y}$ are \emph{order equivalent} 
(denoted as $\orderEquiv{\lst{X}}{\lst{Y}}$), iff 
$\orders{\lst{X}}{\lst{Y}}$ and $\orders{\lst{Y}}{\lst{X}}$.
\end{definition}

\begin{definition} \label{def:order-comp}
\emph{(order compatibility)}
Let $\lst{X}$ and $\lst{Y}$ be lists of attributes where
$\set{X}, \set{Y}\subseteq\R{R}$. 
$\lst{X}$ and $\lst{Y}$ are \emph{order compatible}, 
denoted as $\simular{\lst{X}}{\lst{Y}}$,
iff $\orderEquiv{\lst{XY}}{\lst{YX}}$.
\end{definition}

The order dependency $\orders{\lst{X}}{\lst{Y}}$ means that 
$\lst{Y}$’s values are monotonically
non-decreasing with respect to $\lst{X}$’s values. 
Therefore, if one orders the tuples by $\lst{X}$, they are 
also ordered by $\lst{Y}$.
The order compatibility ($\OC$) $\simular{\lst{X}}{\lst{Y}}$ means that there exists 
a total order
of the tuples in which they are ordered according to both $\lst{X}$ and $\lst{Y}$.

\begin{example}
In Table~\ref{table:taxes}, the $\OD$ 
$\orders{\A{sal}}{\A{taxGrp}}$ holds.
The $\OC$ 
$\simular{\A{taxGrp}}{\A{sal}}$
holds, even though the $\OD$ 
$\orders{\A{taxGrp}}{\A{sal}}$
does not.
\end{example}

$\OD$s have a strong correspondence with $\OC$s and $\FD$s.
An $\OD$ $\orders{\lst{X}}{\lst{Y}}$ holds \emph{iff} 
$\simular{\lst{X}}{\lst{Y}}$ ($\OC$) and 
$\fd{\set{X}}{\set{Y}}$ ($\FD$) hold.
This gives two sources of violations for $\OD$s: \emph{swap}s and \emph{split}s~\cite{od-fundamentals2012}.

\begin{definition}
\emph{(swap)} \label{definition:swap}
A \emph{swap} with respect to
$\OC$ $\oc{\lst{X}}{\lst{Y}}$
is a pair of tuples
$\tup{s}$ and $\tup{t}$
such that
$\tup{s} \prec_{\lst{X}} \tup{t}$
but
$\tup{t} \prec_{\lst{Y}} \tup{s}$.
\end{definition}

\begin{definition}
\emph{(split)} \label{definition:split}
A \emph{split} with respect to
$\FD$ $\fd{\set{X}}{\set{Y}}$
is a pair of tuples
$\tup{s}$ and $\tup{t}$
such that
$\tup{s}[\set{X}] = \tup{t}[\set{X}]$
but
$\tup{s}[\set{Y}] \ne \tup{t}[\set{Y}]$
\end{definition}


\begin{example}
In Table~\ref{table:taxes}, 
given the $\OD$ $\orders{\A{pos}, \A{exp}}{\A{pos}, \A{sal}}$,
tuples $\tup{t}_7$ and $\tup{t}_8$ constitute a swap
(the $\OC$ $\oc{\A{pos}, \A{exp}}{\A{pos}, \A{sal}}$), and
tuples $\tup{t}_6$ and $\tup{t}_7$ constitute a split
(the $\FD$ $\fd{\A{pos}, \A{exp}}{\A{pos}, \A{sal}}$).

\end{example}

\begin{definition} \label{def:equiv-class}
tuples $\tup{s}$ and $\tup{t}$ are \emph{equivalent}
w.r.t.\ set of attributes $\set{X}$ \emph{iff}
$\proj{\tup{s}}{\set{X}}$ $=$ $\proj{\tup{t}}{\set{X}}$.
An attribute set $\set{X}$ partitions tuples into 
\emph{equivalence classes} \cite{tane}. The \emph{equivalence class} 
of tuple $\tup{t}\in\T{r}$ w.r.t.\ $\set{X}$ is denoted by 
$\set{E}(\tup{t}_{\set{X}})$; i.e., $\set{E}(\tup{t}_{\set{X}})$ = 
$\{ \tup{s} \in \T{r}\mid\proj{s}{\set{X}}$ = $\proj{t}{\set{X}} \}$. 
Given a set of attributes $\set{X}$, a \emph{partition} of the table
with respect to $\set{X}$ is the set of all equivalence classes;
i.e., $\Pi_{\set{X}}$ = $\{ \set{E}(\tup{t}_{\set{X}})$ $|$ $\tup{t} \in \T{r} \}$.
\end{definition}


\begin{example}
In Table~\ref{table:taxes}, 
$\set{E}(\proj{\tup{t}[1]}{{\brac{\A{pos}}}})$ = 
$\set{E}(\proj{\tup{t}[2]}{{\brac{\A{pos}}}})$ = 
$\set{E}(\proj{\tup{t}[4]}{{\brac{\A{pos}}}})$ = 
$\{ \tup{t}_1,  \tup{t}_2, \tup{t}_4\}$, and 
$\Pi_{\A{pos}}$ = $\{ \{ \tup{t}_1,  \tup{t}_2, \tup{t}_4\},$ 
$\{ \tup{t}_3, \tup{t}_5, \tup{t}_6, \tup{t}_7, \tup{t}_8\},$
$\{ \tup{t}_9\} \}$.
\end{example}

\subsection{A Canonical Mapping} \label{sec:prelim:canonical}
A natural representation of $\OD$s relies on lists of attributes,
as in the ORDER BY statement in SQL, where the order of 
attributes in the list matters; e.g., 
the $\OD$
$\orders{\A{pos}, \A{sal}}{\A{pos}, \A{exp}}$
is different than
the $\OD$
$\orders{\A{pos}, \A{sal}}{\A{exp}, \A{pos}}$.
This is unlike $\FD$s, where the order of attributes 
does not matter, 
as with
the GROUP BY statement in SQL.
Working within this list-based representation, however, 
has led to discovery frameworks
with factorial worst-case runtimes in the number of attributes \cite{ORDER}.
Fortunately, lists are \emph{not} inherently necessary 
to express $\OD$s. 
In \cite{fastod2017,fastod-vldbj}, the authors rely on a polynomial
mapping of list-based $\OD$s into a logically \emph{equivalent} 
collection of set-based
\emph{canonical} $\OD$s to devise
a discovery framework with exponential worst-case runtime 
in the number of attributes and linear in the number of tuples.


\begin{definition} \label{def:canonical-oc}
\emph{(canonical order compatibility)}
Given a set of attributes $\set{X}$, 
$\simular{\lst{X}'\A{A}}{\lst{X}'\A{B}}$ is 
the $\OC$ that states that attributes $\A{A}$ and $\A{B}$
are \emph{order compatible} within each equivalence class
of $\set{X}$.
We write this as $\simularCtxSet{\set{X}}{\A{A}}{\A{B}}$
in the canonical notation, 
factoring out the common prefix,
and refer to this as a \emph{canonical} $\OC$.
\end{definition}

\begin{definition} \label{def:ofd}
\emph{(order functional dependency)}
Given a set of attributes $\set{X}$, the $\FD$ that states that
an attribute $\A{A}$ is \emph{constant} within each equivalence class
of $\set{X}$ is equivalent to the list-based $\OD$
$\orders{\lst{X}'}{\lst{X}'\A{A}}$. We write this as 
$\ordersCtxSet{\set{X}}{\emptyLst{}}{\A{A}}$
in the canonical notation,
and refer to this as an \emph{order functional dependency} ($\OFD$).
\end{definition}


Given
a canonical $\OC$ of $\simularCtxSet{\set{X}}{\A{A}}{\A{B}}$ or an
$\OFD$ of $\ordersCtxSet{\set{X}}{\emptyLst{}}{\A{A}}$, 
the set $\set{X}$ is referred to as the 
\emph{context} of the respective canonical $\OC$ or $\OFD$.
Intuitively, the context is the common prefix on the left- and right-side
of the corresponding list-based $\OC$ or $\OD$.


Canonical $\OC$s and $\OFD$s  
constitute the canonical $\OD$s; i.e., $\OD \equiv \OC + \OFD$.
The $\OD$ of $\orders{\lst{X}'\A{A}}{\lst{X}'\A{B}}$
is logically equivalent to the 
canonical $\OC$ of $\simularCtxSet{\set{X}}{\A{A}}{\A{B}}$ and
$\OFD$ of $\ordersCtxSet{\set{X}\A{A}}{\emptyLst{}}{\A{B}}$. 
This is 
$\ordersCtxSet{\set{X}}{\A{A}}{\A{B}}$
written in the canonical form.

\begin{example}
In Table~\ref{table:taxes}, 
$\A{sal}$ and $\A{bonus}$ are order compatible
w.r.t.\ the context $\A{pos}$; i.e., 
$\simularC{\A{pos}}{\A{sal}}{\A{bonus}}$.
In the same table, $\A{bonus}$ is constant
w.r.t.\ the context $\A{pos}, \A{sal}$; i.e., 
$\ordersC{\A{pos}, \A{sal}}{\emptyLst{}}{\A{bonus}}$. 
Therefore, $\A{sal}$ orders $\A{bonus}$ w.r.t.\ the context $\A{pos}$; i.e., 
$\ordersCtxSet{\{\A{pos}\}}{\A{sal}}{\A{bonus}}$.
\end{example}

This mapping generalizes:
an $\OD$ $\od{\lst{X}}{\lst{Y}}$ holds
\emph{iff}
\od{\lst{X}}{\lst{X}\lst{Y}}
\emph{and}
\oc{\lst{X}}{\lst{Y}}.   
These can be encoded into an equivalent set
of canonical $\OFD$s and $\OC$s as follows.
In the context of $\set{X}$,
all attributes in $\set{Y}$ must be constants. 
In the context
of all prefixes of $\lst{X}$ and of $\lst{Y}$,
the trailing attributes must be order compatible:

\begin{center}
\(
    \R{R} \models
        \od{\lst{X}}{\lst{X}\lst{Y}}
    \mathrel{\emph{iff}}
    \forall \A{A} \in \lst{Y}.\ 
    \R{R} \models
        \od[\set{X}]{\emptyLst}{\A{A}}
\)
and 
\end{center}


\begin{center}
\(
    \R{R} \models
        \oc{\lst{X}}{\lst{Y}}
    \mathrel{\emph{iff}}
\) 
\(
    \forall i,j.\ 
    \R{R} \models \oc[
                        \Lst{\A{X}_{1}, \ldots, \A{X}_{i-1}}
                        \Lst{\A{Y}_{1}, \ldots, \A{Y}_{j-1}}
                     ]
                     {\A{X}_{i}}
                     {\A{Y}_{j}}
\).
\end{center}

\noindent
Thus,
list-based $\OD$s 
can be polynomially \emph{mapped} to a set of \emph{equivalent} canonical $\OD$s; 
i.e., canonical $\OC$s and $\OFD$s \cite{fastod2017,fastod-vldbj}.

\begin{example}
The $\OD$ 
$\orders{\Lst{\A{A}, \A{B}}}{\Lst{\A{C}, \A{D}}}$
is equivalent to the following canonical $\OD$s:
$\ordersC{\A{A}, \A{B}}{\emptyLst{}}{\A{C}}$, 
$\ordersC{\A{A}, \A{B}}{\emptyLst{}}{\A{D}}$, 
$\simularC{}{\A{A}}{\A{C}}$, 
$\simularC{\A{A}}{\A{B}}{\A{C}}$, 
$\simularC{\A{C}}{\A{A}}{\A{D}}$, and
$\simularC{\A{A}, \A{C}}{\A{B}}{\A{D}}$.
\end{example}


While various algorithms have been proposed for discovering 
$\OD$s, 
most are not \emph{complete}. 
The algorithm described in \cite{ORDER} relies on the 
list-based definition and 
employs aggressive pruning rules to compensate for its 
factorial time complexity, but which make it deliberately incomplete. 
The authors in \cite{od-discovery-jinand2020} 
claim completeness 
but their algorithm misses
$\OD$s in which the same attributes are repeated on the 
left- and right-hand side.
A similar completeness claim has been made in \cite{od-discovery-consonni2019}, 
which is shown to be incorrect in \cite{erratum}.
The set-based $\OD$ discovery
algorithm proposed in \cite{fastod2017}
does offer a 
sound and complete discovery of $\OD$s.
We build our algorithm atop the 
framework introduced in \cite{fastod2017}.

In this work, we refer to canonical $\OC$s simply
as $\OC$s. As will be discussed in Section~\ref{sec:approx}, 
we focus on approximate $\OC$s ($\AOC$s), as an efficient algorithm for validating approximate 
$\OFD$s 
has already been established \cite{tane}. In Section~\ref{sec:approx:optimal}, 
we extend our validation algorithm to handle 
list-based $\OD$s as well.

\subsection{Definition of Approximate ODs} \label{sec:prelim:def}
We define approximate $\OD$s based upon the fewest tuples that must be removed from a 
table for an $\OD$ to hold.
This definition was used for canonical $\OD$s in \cite{fastod2017}; 
their validation step has a quadratic runtime. For
approximate $\FD$s, validation is possible in linear time \cite{tane}.

\begin{definition}
Given a table $\T{r}$ and an $\OD$ 
$\varphi$, 
a set of tuples $\T{s}$ is a \emph{removal set} w.r.t.\ $\varphi$ \emph{iff} 
$\T{r}\setminus \T{s}\models\varphi$.
Let $|\T{r}|$ denote the \emph{cardinality} of $\T{r}$, the number of tuples in $\T{r}$.
A removal set $\T{s}$ is a \emph{minimal} removal set \emph{iff}
it has the smallest 
cardinality over all removal sets; i.e., 
$|\T{s}|=\min(\{|\T{s}|\mid\T{s}\subseteq \T{r},\ \T{r}\setminus \T{s}\models\varphi\})$.
Given $\T{s}$, the \emph{approximation factor} $e(\varphi)$ is defined as 
$|\T{s}|/|\T{r}|$. 
\end{definition}

\begin{example}
Consider Table~\ref{table:taxes} and the $\OC$ of
$\simular{\A{sal}}{\A{tax}}$. 
Here, $\T{s}=\{\tup{t}_1, \tup{t}_2, \tup{t}_4, \tup{t}_6\}$ and
$e(\simular{\A{sal}}{\A{tax}})=4/9\approx0.44$, 
as $\T{r}\setminus\T{s}=\{\tup{t}_3, \tup{t}_5, \tup{t}_7, \tup{t}_8, \tup{t}_9\}$
does not contain any swaps with respect to $\simular{\A{sal}}{\A{tax}}$ 
and no smaller set $\T{s}'$ exists such that 
$\T{r}\setminus\T{s}'\models\simular{\A{sal}}{\A{tax}}$.
\end{example}

Given a table $\T{r}$ and
an approximation threshold $\epsilon$, $0\le\epsilon\le1$, 
the problem of discovering approximate
$\OD$s involves finding all minimal (non-redundant that follow from others) $\OD$s $\varphi$ such that 
$e(\varphi)\le\epsilon$. 
In this work, we focus on the problem of 
validating $\AOD$; i.e., 
verifying whether the approximation factor of a given $\AOD$ is 
less than or equal to a provided threshold. 
We present an optimal algorithm for doing so
and incorporate it into an existing $\OD$ discovery framework.

As discussed in Section~\ref{sec:prelim:canonical}, $\OC$s and $\OFD$s constitute
canonical $\OD$s; i.e., $\OD\equiv\OC+\OFD$. 
There already exists an efficient linear-time algorithm for validating
approximate $\OFD$s, as described in \cite{tane}.
In this work, we present an optimal validation algorithm for $\AOC$s.
Note that when discovering \emph{approximate} $\OC$s and $\OFD$s
given an approximation threshold $\epsilon$, 
$\OD\equiv\OC+\OFD$ does not necessarily hold. 
If approximate $\OC$ 
$\simularCtxSet{\set{X}}{\A{A}}{\A{B}}$ and $\OFD$
$\ordersCtxSet{\set{X}\A{A}}{\emptyLst{}}{\A{B}}$ hold with 
approximation factors $e_1, e_2\le \epsilon$, respectively, it is not 
guaranteed for the corresponding $\OD$ of
$\ordersCtxSet{\set{X}}{\A{A}}{\A{B}}$ to also hold 
approximately with respect to $\epsilon$. 
As to be discussed in Section~\ref{sec:approx:optimal}, however, 
our validation algorithm can easily be extended to validate
list-based approximate $\OD$s as well.

\section{Discovering Approximate OD's} \label{sec:approx}

In Sec.~\ref{sec:approx:framework}, we describe our framework to discover set-based canonical $\OD$s.
In Sec.~\ref{sec:approx:iterative}, we describe the iterative
validation algorithm proposed in \cite{fastod2017,fastod-vldbj}, 
analyze its runtime, and 
provide an example of it failing to find a minimal removal set
and thus overestimating the number of tuples that must be removed.
In Sec.~\ref{sec:approx:optimal}, 
we present our efficient validation algorithm,
based on the longest increasing subsequence (LIS)
problem, analyze its runtime, and prove its minimality and optimality.

\subsection{Discovery Framework} \label{sec:approx:framework}

The algorithm starts the search from singleton sets of attributes
and proceeds to traverse the set-based attribute lattice
in a level-wise manner~\cite{fastod2017,fastod-vldbj}. 
At each level, and when processing the attribute set $\set{X}$, 
the algorithm verifies
$\OC$s of the form
$\simularCtxSet{\set{X}\setminus\{\A{A}, \A{B}\}}{\A{A}}{\A{B}}$
for which $\A{A}, \A{B}\in\set{X}$ and $\A{A}\neq\A{B}$, and
$\OFD$s of the form 
$\ordersCtxSet{\set{X}\setminus\{\A{A}\}}{\emptyLst{}}{\A{A}}$ 
for which $\A{A}\in\set{X}$.

\begin{figure}[t]
    \center
    \scalebox{0.95}{
    \includegraphics[width=0.45\textwidth]{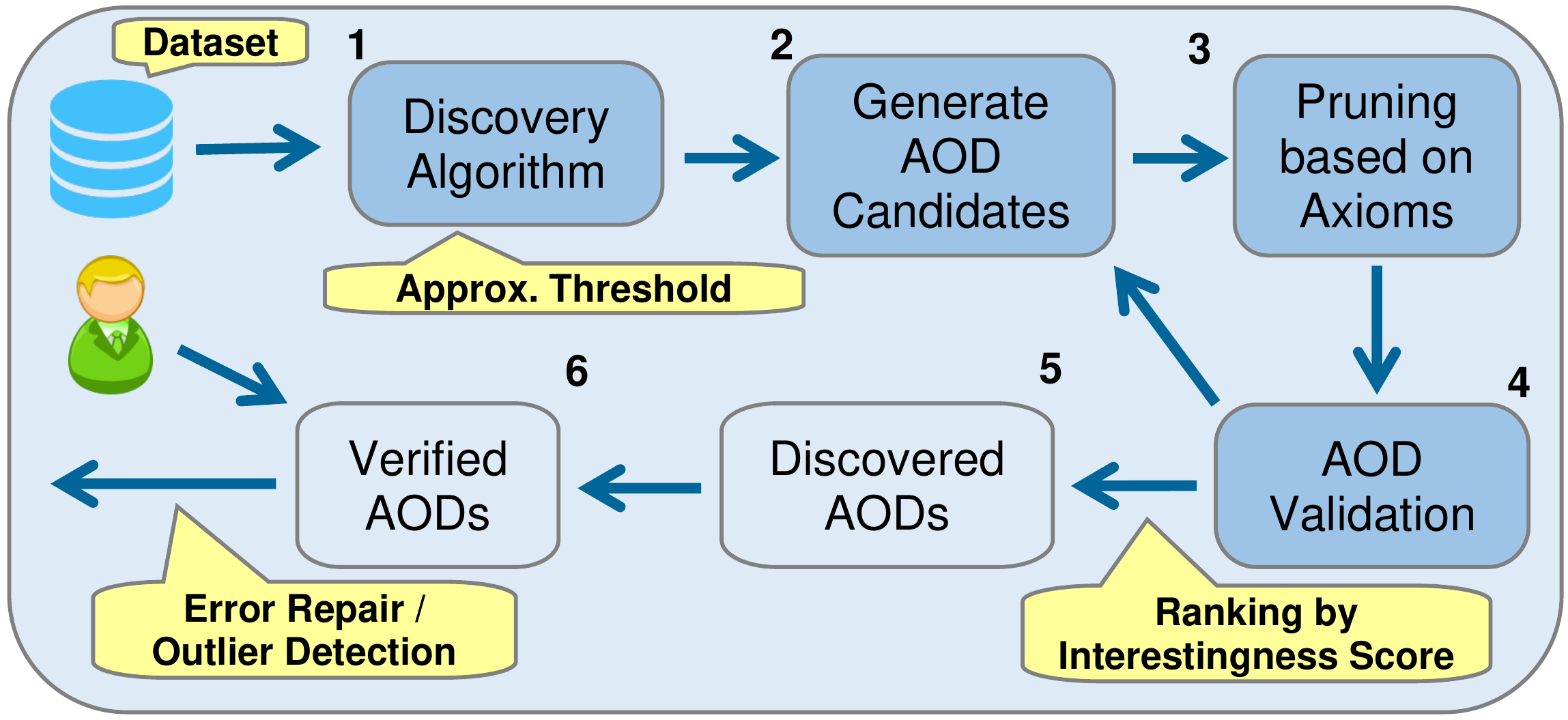}
    }
        \paddingT
     \caption{System framework.}
     \label{fig:framework}
      \paddingD
\end{figure}

Figure~\ref{fig:framework} illustrates the framework.
Candidate $\AOD$s are generated based on 
the attribute sets at the current level of the lattice.
Using the dependencies found in previous levels of the lattice,
these candidates are then pruned by axioms to avoid redundancy~\cite{fastod2017}. 
Our algorithm validates whether each candidate dependency 
holds approximately, given the approximation threshold as input. 
Valid $\AOD$s are then scored and ranked, using the 
measure of interestingness introduced in \cite{fastod2017}.
These discovered $\AOD$s can then 
be manually verified by domain experts, 
to be then used for tasks such as 
error repair or outlier detection, which is an easier task than manual specification.

\subsection{The Iterative Validation Algorithm} \label{sec:approx:iterative}
We first discuss the algorithm described in \cite{fastod2017,fastod-vldbj} 
to validate an approximate $\OC$ given a threshold $\epsilon$.
To validate an $\AOC$, the authors 
propose computing a removal set $\T{s}$ by
iteratively removing a tuple with 
the largest number of swaps, 
which does not guarantee to produce the minimal removal set.
This is repeated until either the $\OC$ holds or the 
number of removed tuples crosses the threshold
$\epsilon|\T{r}|$, in which case the $\AOC$ candidate
is considered invalid.
Note that after removing each tuple, the number of swaps
for the remaining tuples must be updated.

Algorithm~\ref{alg:iterative} validates a candidate
using the iterative approach. 
The steps in Lines \ref{line:iter:sort-by-ab} to
\ref{line:iter:end-remove} are repeated on tuples within
each equivalence class with respect to the context.
Line~\ref{line:iter:count-swaps} uses a variant of merge sort
to count the number of \emph{inversions} in the projection of sorted tuples 
over $\A{B}$, which is equivalent to the number of swaps for each tuple.
Line \ref{line:iter:drop-last} 
removes a tuple with the most swaps  
and Lines \ref{line:iter:begin-swap-for} to \ref{line:iter:end-swap-for}
update the number of swaps for the remaining tuples.
Line~\ref{line:iter:check-invalid} exits if the 
approximation threshold is crossed.

\begin{algorithm}[t]
\begin{small}
\caption{Approx-OC-iterative} 
\label{alg:iterative}
\noindent\textbf{Input:} Table $\T{r}$, $\OC$ $\simularCtxSet{\set{X}}{\A{A}}{\A{B}}$, 
and approximation threshold $\epsilon$. \hspace{1cm} \ \ \ \ \\
\noindent\textbf{Output:} Approximation factor $e$ and removal set $\T{s}$, or ``INVALID'' \ \ \ \ \ \ \ \ \phantom{} \vspace{-0.3cm}
    \begin{algorithmic}[1]
        \STATE $\T{s}=\{\}$
        \FORALL{$\set{E}\in\Pi_{\set{X}}$} \label{line:iter:begin-for}
            \STATE $\T{t}$ $=$ order $\set{E}$ by $[\A{A}$ ASC, $\A{B}$ ASC$]$ \label{line:iter:sort-by-ab}
            \STATE $\proj{\T{t}}{\A{swapCnt}}$ $=$ countInversions$(\proj{\T{t}}{\A{B}})$ \label{line:iter:count-swaps}
            \STATE order $\T{t}$ by $\A{swapCnt}$ ASC \label{line:iter:sort-swaps1}
            \WHILE{$\T{t}$ is not empty} \label{line:iter:begin-remove}
                \STATE $\tup{t}=\T{t}\text{.dropLast()}$ \label{line:iter:drop-last}
                \oneLineIF{$\proj{\tup{t}}{\A{swapCnt}} \mathrel{==} 0$}{\textbf{break}}
                \label{line:iter:check-zero}
                \FORALL{$\tup{s}\in\T{t}$} \label{line:iter:begin-swap-for}
                    \oneLineIF{$\proj{\tup{s}}{\A{A}, \A{B}}$ and $\proj{\tup{t}}{\A{A}, \A{B}}$ are swapped}{$\proj{\tup{s}}{\A{swapCnt}}\mathrel{-}= 1$}
                \ENDFOR \label{line:iter:end-swap-for}
                \STATE order $\T{t}$ by $\A{swapCnt}$ ASC \label{line:iter:sort-swaps2}
                \STATE add $\tup{t}$ to $\T{s}$ \label{line:iter:add-to-s}
                \oneLineIF{$|\T{s}| > \epsilon|\T{r}|$}{\algorithmicreturn\ ``INVALID''} \label{line:iter:check-invalid}
            \ENDWHILE \label{line:iter:end-remove}
        \ENDFOR \label{line:iter:end-for}
        \RETURN $|\T{s}|/|\T{r}|$, $\T{s}$
    \end{algorithmic}
 \end{small}
\end{algorithm}


\begin{example}
Consider Table~\ref{table:taxes} and the $\OC$
$\simular{\A{sal}}{\A{tax}}$. 
Tuple $\tup{t}_7$ has swaps with tuples 
$\tup{t}_1$, $\tup{t}_2$, $\tup{t}_4$, and $\tup{t}_6$,
which is more than any tuple in the table, 
and is thus removed.
In following steps, tuples 
$\tup{t}_5$, $\tup{t}_3$, $\tup{t}_6$, and $\tup{t}_4$
are removed.
Therefore, $\T{s}=\{\tup{t}_3, \tup{t}_4, \tup{t}_5, \tup{t}_6, \tup{t}_7\}$
is reported as a removal set for this $\OC$,
and the 
approximation factor is computed as  $5/9\approx0.56$.
This is larger than the actual approximation 
factor for this $\AOC$; i.e., $0.44$.
\end{example}


Let $m$ denote the number of tuples in an equivalence class.
Lines \ref{line:iter:sort-by-ab} to \ref{line:iter:sort-swaps1}
have runtime $\mathcal{O}(m\log m)$. Lines 
\ref{line:iter:drop-last} to \ref{line:iter:check-invalid}
inside the loop take $\mathcal{O}(m)$ time. Note that since the value of 
$\A{swapCnt}$ for each tuple is bounded by $m$, sorting the tuples in 
Line~\ref{line:iter:sort-swaps2} (as well as Line~\ref{line:iter:sort-swaps1})
can be done in $\mathcal{O}(m)$ time using counting sort. 
In the worst case, 
this loop is repeated $\epsilon n$ times, where $\epsilon$ and $n$ denote
the approximation threshold and the number of tuples in the table, respectively.
Therefore, in the worst case, where $m=n$, 
the runtime of this algorithm is 
$\mathcal{O}(n\log n+\epsilon n^2)$.


\subsection{Our Optimal Validation Algorithm} \label{sec:approx:optimal}
We now present Algorithm~\ref{alg:optimal} based on the 
\emph{longest increasing subsequence} (LIS) problem
to validate an $\AOC$ candidate. 
Lines \ref{line:lis:sort-by-ab} to \ref{line:lis:subtract} are repeated
for the tuples in each equivalence class with respect to the context.
Line~\ref{line:lis:sort-by-ab} orders the tuples
by $[\A{A}, \A{B}]$ in ascending order.
Next, Line~\ref{line:lis:compute-lis} finds a longest
non-decreasing subsequence (LNDS) 
of the projection of tuples over $\A{B}$.
(As $\OC$s are symmetric, we can also sort by $[\A{B}, \A{A}]$ and find a LNDS of projections over $\A{A}$.)
Line~\ref{line:lis:subtract} 
adds the tuples that are \emph{not} in the 
LNDS to the removal set.
Finally, Line~\ref{line:lis:check-thresh} checks whether the $\OC$ holds approximately
with respect to the threshold, and returns the appropriate output.

\begin{algorithm}[t]
\begin{small}
\caption{Approx-OC-optimal} 
\label{alg:optimal}
\noindent\textbf{Input:} Table $\T{r}$, $\OC$ $\simularCtxSet{\set{X}}{\A{A}}{\A{B}}$, 
and approximation threshold $\epsilon$. \hspace{1cm} \ \ \ \ \\
\noindent\textbf{Output:} Approximation factor $e$ and removal set $\T{s}$, or ``INVALID'' \ \ \ \ \ \ \ \ \phantom{} \vspace{-0.3cm}
    \begin{algorithmic}[1]
        \STATE $\T{s}=\{\}$
        \FORALL{$\set{E}\in\Pi_{\set{X}}$}  \label{line:lis:begin-for}
            \STATE $\T{t}$ $=$ order $\set{E}$ by $[\A{A}$ ASC, $\A{B}$ ASC$]$ \label{line:lis:sort-by-ab}
            \STATE $L$ $=$ computeLNDS$(\proj{\T{t}}{\A{B}})$ \label{line:lis:compute-lis}
            \STATE $\T{s}=\T{s}\cup (\proj{\T{t}}{\A{B}}\setminus L)$ \label{line:lis:subtract}
        \ENDFOR \label{line:lis:end-for}
        \oneLineIFandElse{$|\T{s}| \le \epsilon|\T{r}|$}{\algorithmicreturn\ $|\T{s}|/|\T{r}|$, $\T{s}$}{\algorithmicreturn\ ``INVALID''} \label{line:lis:check-thresh}
    \end{algorithmic}
 \end{small}
\end{algorithm}

\begin{example}
Consider Table~\ref{table:taxes} and the $\OD$
$\simular{\A{sal}}{\A{tax}}$. 
After ordering the tuples according to $\A{sal}$ and breaking ties by $\A{tax}$, 
the projection of the tuples over $\A{tax}$ is the list 
$[2K,$ $2.5K,$ $0.3K,$ $12K,$ $1.5K,$ $16.5K,$ $1.8K,$ $7.2K,$ $16K]$. 
The LNDS of this list is 
$[0.3K,$ $1.5K,$ $1.8K,$ $7.2K,$ $16K]$ and thus, the removal set is
$\T{s}=\{\tup{t}_1, \tup{t}_2, \tup{t}_4, \tup{t}_6\}$.
Thus, the approximation factor is $4/9\approx0.44$.
\end{example}

Again, let $m$ denote the number of tuples in an equivalence class. 
Sorting the tuples in each equivalence class takes 
$\mathcal{O}(m\log m)$ time (Line~\ref{line:lis:sort-by-ab}).
To compute a 
LNDS
of a list with length $m$, a dynamic programming algorithm 
from \cite{LIS-lower-bound-fredman1975}
with small modifications and
with runtime ${\mathcal O}(m\log m)$ 
is employed
(Line~\ref{line:lis:compute-lis}).
In Line~\ref{line:lis:subtract}, since $L$ is a subsequence of $\proj{\T{t}}{\A{B}}$, 
$\proj{\T{t}}{\A{B}}\setminus L$ can be computed in $\mathcal{O}(m)$ time by traversing 
both lists once. Therefore, the worst case runtime of this
algorithm, which occurs when $m=n$, is $\mathcal{O}(n\log n)$.


We now prove minimality and optimality of our algorithm.

\begin{theorem} \label{thm:oc-correctness}
The set $\T{s}$ generated using Algorithm~\ref{alg:optimal}
is a minimal removal set with respect to the given $\AOC$.
\end{theorem}

\begin{theorem} \label{thm:oc-optimality}
Algorithm~\ref{alg:optimal} has the optimal runtime for
validating an $\AOC$ candidate.
\end{theorem}

Our validation algorithm easily extends to approximate 
$\OD$s of the form $\ordersCtxSet{\set{X}}{\A{A}}{\A{B}}$. 
We again use Algorithm~\ref{alg:optimal}, but in Line~\ref{line:lis:sort-by-ab}, tuples are ordered 
according to the ascending order 
over $\A{A}$, but ties are broken according 
to the \emph{descending} order over $\A{B}$.
Intuitively, this forces the solution to the LNDS problem 
in Algorithm~\ref{alg:optimal} to 
remove all \emph{split}s in the table 
(removal of \emph{swap}s is already ensured 
similar to Algorithm~\ref{alg:optimal} for approximate $\OC$s).
\footnote{This idea can be extended to list-based $\OD$s of 
the form $\orders{\lst{X}}{\lst{Y}}$, by ordering tuples 
in ascending order of $\lst{X}$ and breaking ties using the
descending order over $\lst{Y}$.
}

\newcounter{exp-c}
\setcounter{exp-c}{0}

\section{Experiments} \label{sec:exp}


We implemented our approximate $\OC$ validation algorithm on top of 
a Java implementation of the set-based $\OD$ discovery 
framework from \cite{fastod2017}.
We implemented our new LIS-based algorithm as well as 
the iterative algorithm 
using the same technologies to ensure that the improvements in runtime 
are not due to implementation differences. 
Unless mentioned otherwise, we set the approximation threshold to 10\% 
and use ten attributes.
We run our experiments on a machine with Xeon CPU 2.4GHz with 64GB RAM, and 
use datasets from the Bureau of Transportation Statistics
and the North Carolina State Board of Elections:

\begin{enumerate}
    \item \textbf{flight} contains information
    such as date, origin, destination, and airline about flights in the 
    United States and has 1M tuples and 35 attributes (\url{https://www.bts.gov}).
    \item \textbf{ncvoter} contains information such as 
    registration number, age, and address about voters 
    in North Carolina and has 5M tuples and 30 attributes (\url{https://www.ncsbe.gov}).
\end{enumerate}

\subsection{Scalability} \label{sec:exp:scal}
\refstepcounter{exp-c}\label{exp:scal-tup}
\textbf{Exp-1: Scalability in $|\T{r}|$.}
We measure the runtime (in seconds) of the $\AOD$ discovery framework 
that uses our validation algorithm by varying the number of tuples in our datasets, as reported in 
Figure~\ref{fig:scalrow}. For now, ignore the curves 
labeled ``$\OD$'' and ``$\AOD$ (iterative)'', as well as the numbers
next to the datapoints. 
The $\AOD$ discovery framework implemented using 
our optimal algorithm scales up to millions of tuples. 

\begin{figure}[t]
    \center
    \includegraphics[width=0.45\textwidth]{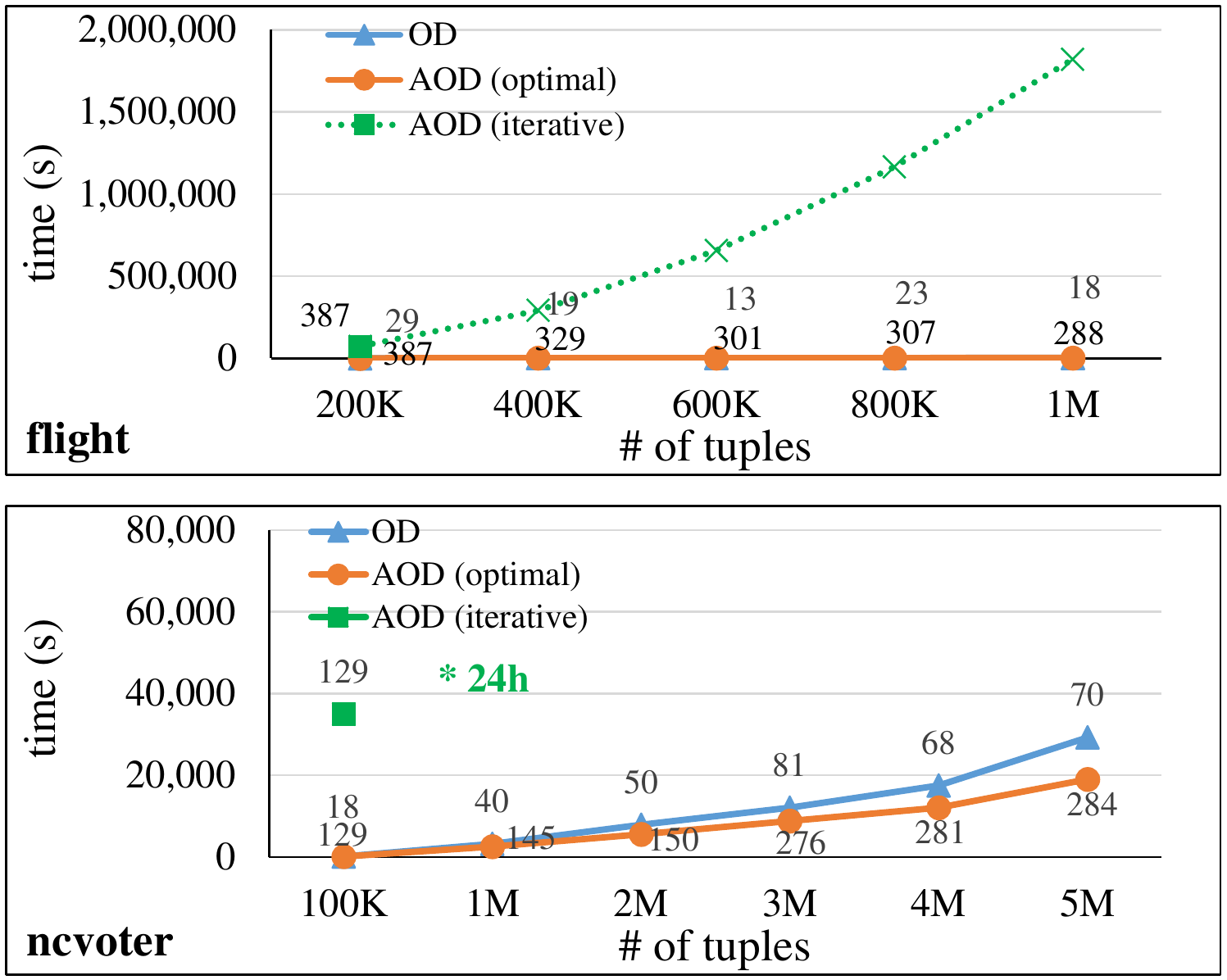}
    
        \paddingT
     \caption{Scalability in $|\textbf{r}|$.}
     \label{fig:scalrow}
      \paddingD
      \vspace{+0.1 in}
\end{figure}

\refstepcounter{exp-c}\label{exp:scal-attr}
\noindent\textbf{Exp-2: Scalability in $|\T{R}|$.}
Next, we measure the runtime of the discovery framework
in milliseconds, as illustrated in Figure~\ref{fig:scalattr}. 
We use 1K tuples of our datasets (to allow experiments
with a large number of attributes in reasonable time) 
and vary the number of attributes 
in multiples of five.
In this experiment, the runtime has an exponential growth
(the Y-axis in Figure~\ref{fig:scalattr} is in log scale). 
This is expected since the number of 
$\OD$s increases exponentially with the number of tuples. 
The higher runtime on $\sf{ncvoter}$ compared to $\sf{flight}$ is attributed to 
having more $\OD$s in higher levels of the lattice 
(with larger contexts).

\begin{figure}[t]
    \center
    \includegraphics[width=0.45\textwidth]{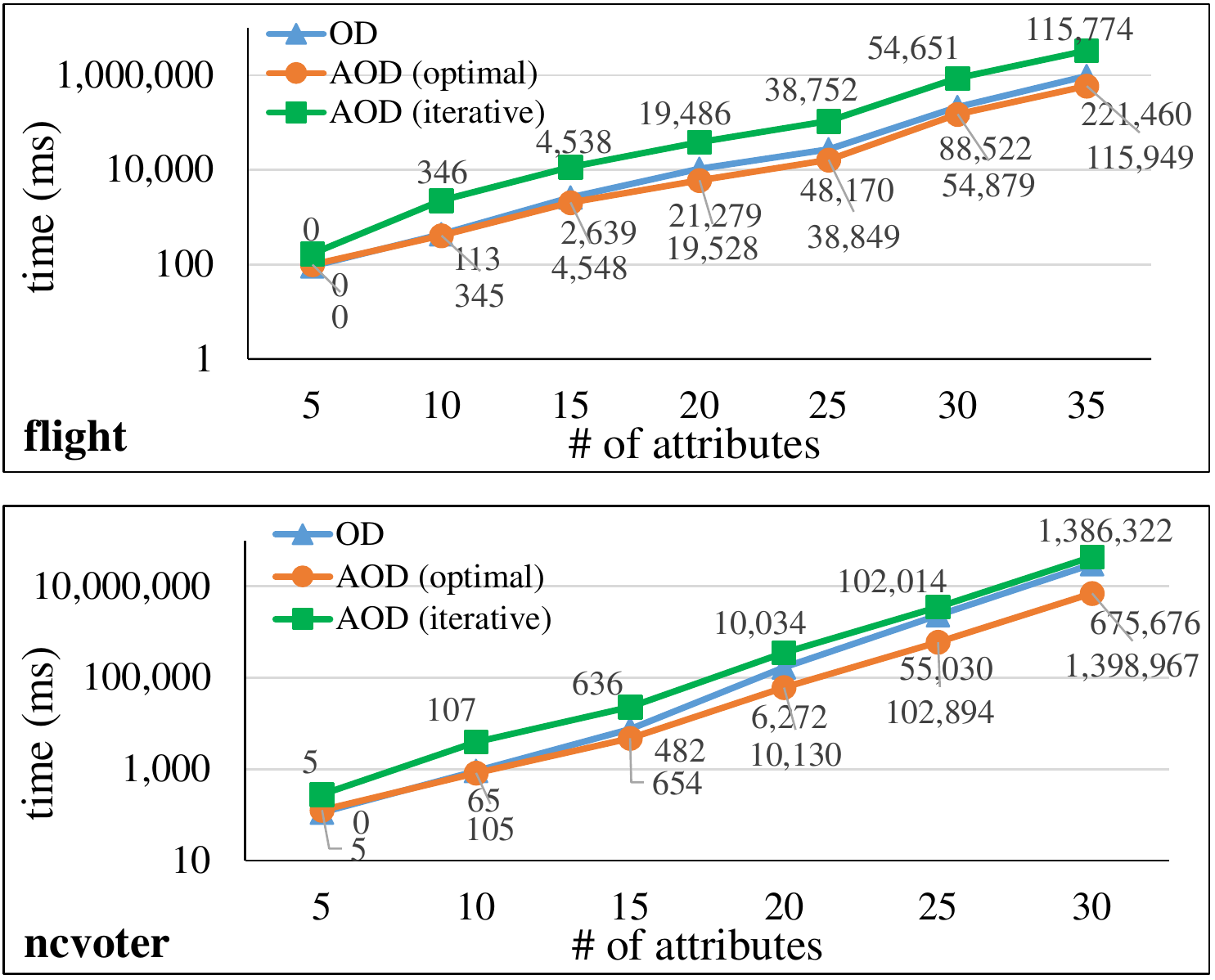}
    
        \paddingT
     \caption{Scalability in $|\textbf{R}|$.}
     \label{fig:scalattr}
      \paddingD
\end{figure}

\subsection{Comparison with the Iterative Algorithm} \label{sec:exp:comp-iter}

\refstepcounter{exp-c}\label{exp:vs-iter-time}
\textbf{Exp-3: Runtime comparison with the iterative algorithm.}
As discussed in Section~\ref{sec:approx}, our $\AOC$ validation algorithm has
time complexity $\mathcal{O}(n\log n)$, while the iterative algorithm 
proposed in \cite{fastod2017,fastod-vldbj} has time complexity 
$\mathcal{O}(n\log n+\epsilon n^2)$. 
Figures \ref{fig:scalrow}, \ref{fig:scalattr}, and \ref{fig:vsperc}
illustrate the running times
of the $\AOD$ discovery framework when using these 
two validation algorithms. 

As shown in Figure~\ref{fig:scalrow}, while when using our algorithm, the framework can 
discover $\AOC$s in datasets with up to millions of tuples, when using the 
iterative algorithm, it does not terminate within 24 hours on 
400K and 1M tuples of the $\sf{flight}$ and $\sf{ncvoter}$ datasets, respectively
(the running times for the $\sf{flight}$ dataset have been projected for 
better comparison). In cases where the framework 
equipped with the iterative algorithm 
terminates within the time limit, it is orders of magnitude slower. 
In Figure~\ref{fig:scalattr}, while the differences are not as 
pronounced (as the number of tuples is too small), 
using our validation algorithm still makes the framework
almost an order of magnitude faster. 

We next experiment with the approximation threshold, by
using 10K tuples from our datasets and setting
the approximation threshold to 0, 5, 10, 15, 20, and 25 percent.
As Figure~\ref{fig:vsperc} illustrates,
while a larger approximation threshold does not increase
the runtime of our algorithm
(the runtime decreases in some cases
due to better pruning opportunities), it increases
the runtime of the iterative approach at an almost linear rate. 
This aligns with the time complexity of 
these algorithms, as analyzed in Section~\ref{sec:approx}.

As mentioned in Section~\ref{sec:intro}, validating $\AOC$s
becomes the bottleneck of the $\AOD$ discovery framework 
when using the iterative algorithm. This is verified in our experiments, 
as up to 99.6\% 
of the total runtime is 
spent on validation.
Using our LIS-based validation algorithm, we reduce 
the time spent on validating $\AOC$s by up to 99.8\%, 
which results in the orders-of-magnitude improvement in runtime 
discussed before.


\begin{figure}[t]
    \center
    \includegraphics[width=0.45\textwidth]{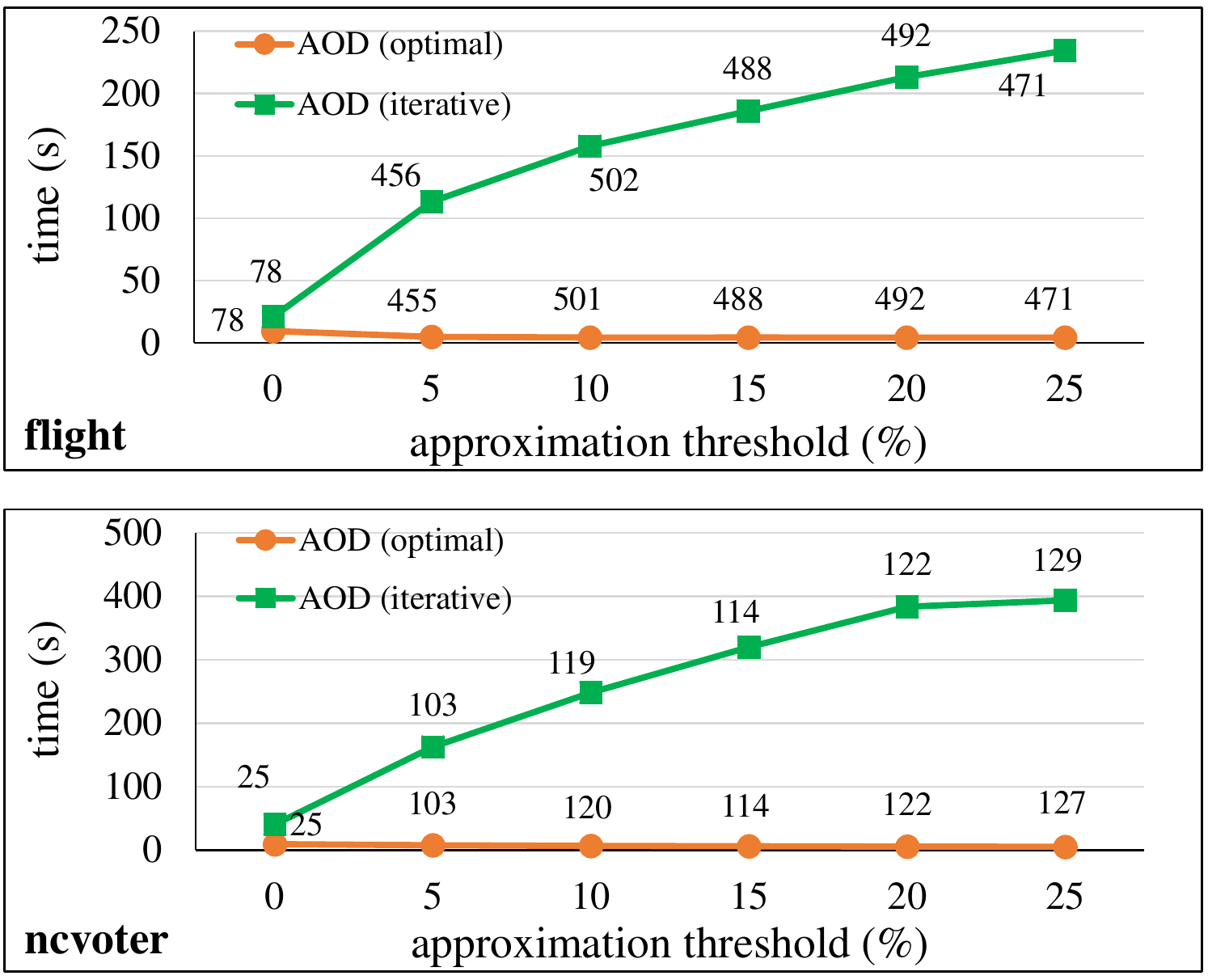}
    
        \paddingT
     \caption{The effect of the approximation threshold.}
     \label{fig:vsperc}
      \paddingD
\end{figure}

\refstepcounter{exp-c}\label{exp:vs-iter-aoc}
\noindent\textbf{Exp-4: Removal sets and validating $\AOC$s using the iterative algorithm.}
While our validation algorithm guarantees finding a minimal
removal set 
for a given $\OC$ (as is proved in Section~\ref{sec:approx:optimal}), 
the iterative algorithm may \emph{overestimate} the size of a minimal removal set. 
This results in removal sets which are on average around 1\% larger 
than the true minimal removal set. 

Overestimating the approximation factor may result in
missing valid $\AOC$s, if the true approximation factor 
is close to the input threshold.
In Figures \ref{fig:scalrow}, \ref{fig:scalattr}, and \ref{fig:vsperc}, 
the numbers inside the plots indicate the number of $\OC$s or $\AOC$s found
by an algorithm.
We have not listed the number of approximate $\OFD$s since 
this work focuses on discovering $\AOC$s. 
(Wherever the plots for our algorithm and the algorithm for exact $\OD$s
overlap, the numbers on the bottom correspond to our approach.)
In our experiments, the iterative approach
misses up to 2\% of the valid $\AOC$s 
found using our optimal approach. 

Missing these $\AOC$s could 
have potentially
severe consequen\-ces.
For instance, in the $\sf{flight}$ dataset, the $\AOC$ of 
$\simular{\A{arrivalDelay}}{\A{lateAircraftDelay}}$ holds with 
an approximation factor of 9.5\%. This $\AOC$ 
points out that generally, delays in arrival are due to 
the aircraft and not other causes; e.g., security or weather delays. 
However, the iterative algorithm overestimates the approximation factor as
10.5\%. This results in the framework missing this valid $\AOC$
when using an approximation threshold of 10\%.
Note that missing some $\AOC$s results in 
different pruning opportunities, and, as a result,
the set of discovered $\AOC$s, which explains why the 
iterative algorithm discovers more $\AOC$s in some cases.

Furthermore, as has been discussed for Exp-\ref{exp:vs-iter-time},
the running time of the iterative algorithm on larger 
datasets is prohibitively long.
On such datasets, using the iterative algorithm 
results in missing \emph{all} valid $\AOC$s.
For instance, in the $\sf{ncvoter}$ dataset with 5M tuples and with 
the approximation threshold set to 20\%, the $\AOC$ of
$\simular{\A{municipalityAbbrv}}{\A{municipalityDesc}}$
is discovered, which points to exceptions 
in creating abbreviations for municipalities; 
e.g., ``Raleigh'' is abbreviated as ``RAL'', while 
``Charlotte'' is abbreviated as ``CLT''.
However, this $\AOC$ does not hold in our 100K 
sample of tuples when using this threshold. 
Therefore, this dependency would have been missed 
by using the iterative validation algorithm, as it 
exceeds the time limit on the full dataset. 

\subsection{Comparison with 
Exact OD Discovery} \label{sec:exp:comp-exact}

\refstepcounter{exp-c}\label{exp:vs-exact-time}
\noindent\textbf{Exp-5: Lattice level of $\AOC$s and runtime improvements.}
$\AOC$s tend to reside
in lower levels of 
the lattice (with smaller contexts).
In our scalability experiments in the number of tuples (Exp-\ref{exp:scal-tup}), 
the $\AOC$s are on average 
$1.2$ levels lower on the lattice.
Similarly, in experiments in the number of attributes (Exp-\ref{exp:scal-attr}), 
the $\AOC$s are on average 
$0.5$ levels lower on the lattice.
Figure~\ref{fig:lattice-ncv} shows the number of $\OC$s or $\AOC$s found
at each level of the lattice, when using 5M tuples and 10 attributes
of the $\sf{ncvoter}$ dataset. On this dataset, 
the average lattice level of the discovered dependencies drops
from $5.6$ to $4.3$ when using our approximate algorithm.
As discussed in \cite{fastod2017} and \cite{fastod-vldbj}, 
dependencies found in lower levels of the lattice are likely to be more interesting.

\balance

Furthermore, as discussed in Section~\ref{sec:approx:framework}, 
our discovery framework first validates candidates on 
lower levels of the lattice, and then applies pruning rules
to generate the candidates on higher levels of the lattice.
Therefore, by finding $\AOC$s in lower levels, 
the algorithm can use pruning rules more effectively earlier 
in the discovery process,
resulting in pruning some candidates on higher levels of the lattice. 
The effects of such pruning opportunities are not noticed when 
using the iterative validation algorithm, 
due to its prohibitively long running time.
However, we optimally reduce the runtime of the 
validation step, resulting in runtime improvements
for the discovery framework.

Figures \ref{fig:scalrow} and \ref{fig:scalattr} show the 
running times of the algorithms for discovering exact and approximate $\OD$s.
Even though validation of $\AOC$s has a worse
runtime compared to exact $\OC$s, 
i.e., $\mathcal{O}(n\log n)$, 
as opposed to $\mathcal{O}(n)$, 
due to the extra pruning opportunities described above, 
the total runtime of the discovery framework for $\AOD$s
can even be lower than the discovery framework for exact $\OD$s;
i.e., up to 34\% and 76\% faster in experiments in the number of tuples 
and attributes, respectively. 
The pronounced effect in the experiments in the number 
of attributes is due to having 
a smaller number of tuples. 

\begin{figure}[t]
    \center
    \includegraphics[width=0.45\textwidth]{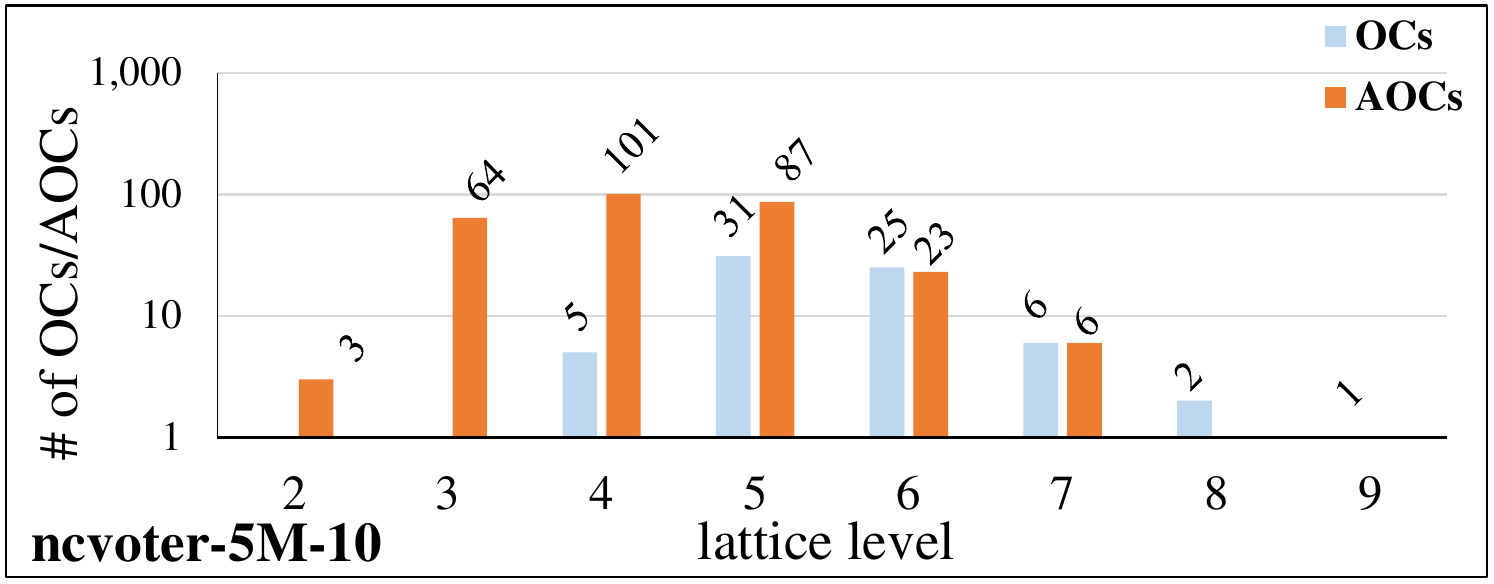}
        \paddingT
     \caption{Number of discovered OCs/AOCs in each level.}
     \label{fig:lattice-ncv}
      \paddingD
      \vspace{+0.1 in}
\end{figure}

\refstepcounter{exp-c}\label{exp:vs-exact-aoc}
\noindent\textbf{Exp-6: Discovered $\AOC$s compared to $\OC$s.}
The exact algorithm fails to discover meaningful $\OC$s in presence of anomalies, 
or even if a single value is erroneous. 
However, valid $\AOC$s may hold in
such instances. 
Other than the $\AOC$s discussed 
in Exp-\ref{exp:vs-iter-aoc}, in the $\sf{flight}$ dataset,
we discover the $\AOC$ 
$\simular{\A{originAirport}}{\A{IATACode}}$
with an 8\% approximation factor.
This $\AOC$ can be used to identify data quality issues, 
as the airport identifier must uniquely correspond to the 
IATA code in ascending order.
Furthermore, the $\AOC$
$\simular{\A{streetAddress}}{\A{mailAddress}}$
holds in the $\sf{ncvoter}$ dataset with an approximation factor of 18\%. 
This $\AOC$ can point to exceptions in address formats.

As shown in Figures \ref{fig:scalrow} and \ref{fig:scalattr}, 
by discovering $\AOC$s, we can find more dependencies in the data.
Even if there are fewer $\AOC$s than 
$\OC$s (e.g., the $\sf{flight}$ dataset in Exp-\ref{exp:scal-attr}), 
the discovered dependencies are on lower levels of the lattice, 
as shown in Exp-\ref{exp:vs-exact-time}, 
which makes them more interesting \cite{fastod2017,fastod-vldbj}.
If the number of discovered dependencies 
is too large, the interestingness measure proposed in \cite{fastod-vldbj}
can be used to rank the $\AOC$s.
In fact, the example $\AOC$s 
that we have identified in Exp-\ref{exp:vs-iter-aoc} and 
in this experiment, were all ranked as the most
interesting $\AOC$s based on this measure. 

\section{Conclusions} \label{sec:conclusions}
We proposed a new validation algorithm 
for approximate $\OD$s 
and proved its minimality and runtime optimality. We then implemented
our approach in an existing canonical $\OD$ discovery framework and 
demonstrated significant gains compared to existing 
frameworks for discovering exact and approximate $\OC$s.
In future work, we will study new approaches for 
discovering approximate $\OC$s, such as hybrid sampling, as 
done in \cite{hybrid} for $\FD$s.
We will also extend our approximate 
$\OC$ discovery framework to distributed settings, 
similar to the work in \cite{distributed-dependency-discover}.

\bibliographystyle{ACM-Reference-Format}
\bibliography{bibliography}

%
% ...

\newpage
\section{Appendix} \label{sec:appendix}

\begin{theorem} \label{thm:oc-correctness2}
The set $\T{s}$ generated using Algorithm~\ref{alg:optimal}
is a minimal removal set with respect to the given $\AOC$.
\end{theorem}

\begin{proof} \label{proof:oc-correctness}
Since in the $\AOC$ validation problem 
the tuples within different partition groups with respect to the 
context are independent of each other. Without loss of generality, 
assume the $\OC$ candidate has an empty context and let
$\simular{\A{A}}{\A{B}}$ denote it.
Let list $B$ denote the projection of tuples over $\A{B}$ after 
ordering them by $\A{A}$ and breaking ties by $\A{B}$,
and $L$ denote a 
LNDS of $B$.
Let $L(\tup{s})$ and $l_i$ denote the index of tuple 
$\tup{s}$ in $L$ (assuming $\tup{s}$ is in $L$) and the 
$i$-th element in $L$, respectively.
Finally, let $\T{s}$ be the minimal removal set 
found using our algorithm; i.e., the set of tuples not in $L$. 

First, we prove that 
$\T{s}$ is a \emph{removal set}. 
Assume $\T{r}\setminus\T{s}\not\models\varphi$. 
Thus, there exist tuples $\tup{s}, \tup{t}\in\T{r}\setminus\T{s}$ 
such that
$\proj{\tup{s}}{\A{A}}<\proj{\tup{t}}{\A{A}}$ and
$\proj{\tup{t}}{\A{B}}<\proj{\tup{s}}{\A{B}}$ (a swap).
Since $\tup{s}, \tup{t}\in \T{r}\setminus\T{s}$, 
$\proj{\tup{s}}{\A{B}}$ and $\proj{\tup{t}}{\A{B}}$ are in $L$.
Moreover, because $\proj{\tup{s}}{\A{A}}<\proj{\tup{t}}{\A{A}}$, 
$\tup{s}$ is before $\tup{t}$ in the ordering and
$L(\tup{s})<L(\tup{t})$. 
However, this means that $L$ is not nondecreasing, as
$\proj{\tup{t}}{\A{B}}<\proj{\tup{s}}{\A{B}}$, 
which contradicts the assumption that $L$ 
is a LNDS of $B$. 
Therefore, $\T{r}\setminus\T{s}\models\varphi$.

Next, we prove that among all sets $\T{t}$ such that 
$\T{r}\setminus\T{t}\models\varphi$, $\T{s}$ has the smallest cardinality; 
i.e., it is a \emph{minimal} removal set.
Assume that a set $\T{s}'$ with $|\T{s}'|<|\T{s}|$ exists,
such that $\T{r}\setminus\T{s}'\models\varphi$. 
Construct a subsequence $L'$ of $B$ by including the projection of tuples
that do \emph{not} exist in $\T{s}'$. 
Because $|\T{s}'|<|\T{s}|$, $L'$ is longer than $L$.
Since $\T{r}\setminus\T{s}'\models\varphi$, 
there do not exist tuples
$\tup{s}, \tup{t}\in\T{r}\setminus\T{s}'$ such that
$\proj{\tup{s}}{\A{A}}<\proj{\tup{t}}{\A{A}}$ and
$\proj{\tup{t}}{\A{B}}<\proj{\tup{s}}{\A{B}}$.
Furthermore, since the tuples are order by $\A{A}$ and 
ties are broken by $\A{B}$, 
$L(\tup{s})<L(\tup{t})$ for all 
$\tup{s}, \tup{t}\in\T{r}\setminus\T{s}'$ such that 
$\proj{\tup{s}}{\A{A}}=\proj{\tup{t}}{\A{A}}$ and
$\proj{\tup{s}}{\A{B}}<\proj{\tup{t}}{\A{B}}$.
Therefore, there do not exist $i<j$ such that
$l'_j<l'_i$. Thus, $L'$ is a nondecreasing subsequence of 
$B$ which is longer than $L$. 
This is in contradiction with $L$ being a 
LNDS of $B$. Therefore, $\T{s}$ has the smallest 
size possible for a removal set.
\end{proof}

\begin{theorem} \label{thm:oc-optimality2}
Algorithm~\ref{alg:optimal} has the optimal runtime for
validating an $\AOC$ candidate.
\end{theorem}

\begin{proof} \label{proof:oc-optimality}
In \cite{LIS-lower-bound-fredman1975}, the author has proved an $\Omega(n\log n)$ lower bound for a decision variant of the LIS problem, here referred to as LIS-DEC, as follows: given a list of $n$ distinct values, is the length of a longest increasing subsequence, denoted by $|L|$, larger than or equal to $k=\lfloor 3n^{1/2}\rfloor$? 
To prove the same lower bound for the $\AOC$ validation problem (and thus, the optimality of our algorithm), we offer a linear-time mapping from instances of LIS-DEC to $\AOC$ validation instances, in which $|L|\ge k$ \emph{iff} the $\AOC$ instance is valid with an approximation threshold of $1-k/n$.

Let $B$ be the input list for a LIS-DEC instance, consisting of $[b_1,$ $b_2,$ $\dots,$ $b_n]$. For the corresponding $\AOC$ instance, let $\T{r}$ be a table with attributes $\A{A}$ and $\A{B}$, and consider the $\OC$ $\simular{\A{A}}{\A{B}}$, denoted by $\varphi$. For each $b_i\in B$, add the tuple $\tup{t}_i$ $=$ $(i, b_i)$ to  $\T{r}$.

We first prove that if the $\AOC$ candidate is valid, then the LIS-DEC instance holds; i.e., if $e(\varphi)\le 1-k/n$, then $|L|\ge k$. Let $\T{s}$ be a minimal removal set with respect to $\simular{\A{A}}{\A{B}}$. Since the $\AOC$ candidate is valid, $|\T{s}|\le n-k$. We now create $L$, a subsequence of $B$, which contains all $b_i$'s where $\tup{t}_i\not\in\T{s}$. Since $|L|=n-|\T{s}|$, we know that $|L|\ge k$. Therefore, we only need to prove that $L$ is an increasing subsequence of $B$. Assume $L$ is not increasing; therefore, there exist $b_i, b_j\in L$ such that $i<j$ and $b_j<b_i$. Since $b_i, b_j\in L$, $\tup{t}_i, \tup{t}_j\in \T{r}\setminus\T{s}$.
However, $\proj{\tup{t}[i]}{\A{A}}<\proj{\tup{t}[j]}{\A{A}}$, while $\proj{\tup{t}[j]}{\A{B}}<\proj{\tup{t}[i]}{\A{B}}$. Thus, there exists a swap with respect to $\simular{\A{A}}{\A{B}}$ in  $\T{r}\setminus\T{s}$, which contradicts the assumption that $\T{s}$ is a removal set w.r.t.\  $\simular{\A{A}}{\A{B}}$.
Therefore, $L$ is an increasing subsequence of $B$ and the LIS-DEC instance holds.

We now prove the other direction, i.e., if the LIS-DEC holds, then the $\AOC$ candidate is valid. Let $L$ be a longest increasing subsequence of $B$, where $|L|\ge k$. We construct a set of tuples $\T{s}$, by including all tuples  $\tup{t}_i\in\T{r}$, where $b_i\not\in L$. We know that $|\T{s}|/n\le 1-k/n$ since $|\T{s}|=n-|L|$. Therefore, to prove that the $\AOC$ candidate is valid,  we only need show that $\T{s}$ is a removal set with respect to $\T{t}$ and $\varphi$; i.e., $\T{r}\setminus\T{s}\models\varphi$. Assume $\T{s}$ is not a removal set. Thus, there exist tuples $\tup{t}[i], \tup{t}[j]\in\T{r}\setminus\T{s}$ such that $i<j$  (and therefore, $\proj{\tup{t}[i]}{\A{A}}<\proj{\tup{t}[j]}{\A{A}}$), but $\proj{\tup{t}[j]}{\A{B}}<\proj{\tup{t}[i]}{\A{B}}$. Since $\tup{t}[i], \tup{t}[j]\in \T{r}\setminus\T{s}$,  $b_i$ and $b_j$ are in $L$. Therefore, there exist values $b_i$ and $b_j$ in $L$ such that $i<j$, but $b_j<b_i$. Thus, $L$ is not increasing,  which is in contradiction with the assumption that $L$ is a LIS of $B$. 
Therefore, $\T{r}\setminus\T{s}\models\varphi$ and the $\AOC$ candidate is valid.

Therefore, the LIS-DEC instance holds if and only if the corresponding $\AOC$ candidate is valid. Since the mapping takes linear time in the size of the input, $\Omega(n\log n)$ is also a lower bound for the $\AOC$ validation problem.
\end{proof}

\end{document}